\documentclass[10pt]{article}
\usepackage{amsmath, amsthm, amssymb}
\usepackage{graphicx}
\usepackage{fullpage}
\usepackage{setspace} 
\usepackage{latexsym}
\usepackage{url,hyperref}

\newtheorem{theorem}{Theorem}[section]
\newtheorem{corollary}[theorem]{Corollary}
\newtheorem{lemma}[theorem]{Lemma}
\newtheorem{proposition}[theorem]{Proposition}
\newtheorem{definition}[theorem]{Definition}

\newtheorem{claim}[theorem]{Claim}
\newtheorem{observation}[theorem]{Observation}
\newcommand{\ignore}[1]{}
\def \qed {\hspace*{0pt} \hfill {\quad \vrule height 1ex width 1ex depth 0pt}
 \medskip}
\def \Lyap {\rm L}

\newcommand{\R}{\ensuremath{\mathbb R}}

\newtheorem{fact}[theorem]{Fact}

\begin{document}

\title{Ascending Auctions and Walrasian Equilibrium}
\author{Oren Ben-Zwi\thanks{University of Vienna Faculty of Computer Science Vienna, Austria {\tt oren.ben-zwi@univie.ac.at}}
\and Ron Lavi\thanks{Technion -- Israel Institute of Technology  Faculty of Industrial Engineering and Management {\tt ronlavi@ie.technion.ac.il}}
\and Ilan Newman\thanks{University of Haifa Department of Computer Science {\tt ilan@cs.haifa.ac.il}}}

\maketitle

\begin{abstract}
We present a family of submodular valuation classes that generalizes gross substitute. We show that Walrasian equilibrium always exist for one class in this family, and there is a natural  ascending auction which finds it.  We prove some new structural properties on gross-substitute auctions which, in turn, show that the known ascending auctions for this class (Gul-Stacchetti and Ausbel) are, in fact, identical. We generalize these two auctions, and provide a simple proof that they terminate in a Walrasian equilibrium.
\end{abstract}


\section{Introduction}

Ascending auctions are an important auction format with many advantages such as visibility of price formation and gradual disclosure of valuations. It is harder for a dishonest auctioneer to manipulate the bids in such auctions, and it is easier to adjust such auctions to a setting with budget constraints. It is a popular auction format in practice, and its theoretical properties have been widely studied.

A significant part of the literature tries to characterize settings where the end outcome of the
ascending auction is a Walrasian equilibrium: this is an outcome where each bidder is being allocated her most preferred subset of items, under final prices, and, furthermore, all items are being allocated. Such an outcome is desirable as it provides some stability, and additionally it guarantees optimal social welfare. Kelso and Crawford~\cite{KelsoCr1982} define the class of ``gross substitutes'' (GS) valuations, a small subclass of submodular valuations that ensures the existence of a Walrasian equilibrium. They argue that a natural ascending auction terminates in a Walrasian equilibrium, when valuations are GS and bidders report true demands throughout the auction.

A complete treatment of ascending auctions for gross substitutes valuations is offered by Gul and Stacchetti~\cite{GulSt2000}. They give a formal, specific ascending auction and prove that it terminates in a Walrasian equilibrium when valuations are GS. They additionally show in a companion paper~\cite{GulSt1999} that GS is the unique maximal class of valuations that contains unit-demand valuations and guarantees the existence of a Walrasian equilibrium: for any non-GS valuation there exist unit-demand valuations such that a Walrasian equilibrium does not exist for these unit-demand valuations coupled with the given non-GS valuation.

In this paper we study whether it may be possible to construct an ascending auction that reaches a Walrasian equilibrium for submodular valuations that are not GS. This question is twofold: First, what are natural subclasses of submodular valuations for which a Walrasian equilibrium is guaranteed to exist? Second, for these valuation classes, can we reach a Walrasian equilibrium via an ascending auction? We should remark that given the existing literature on ascending auctions described above, this question may initially seem hopeless. After all, GS is the maximal class that ensures existence of Walrasian equilibrium. However, in this paper we demonstrate that this intuition is false. Our main result shows another natural subclass of submodular valuations that ensures the existence of a Walrasian equilibrium, and gives an ascending auction that finds this equilibrium outcome.

\vspace{2mm}

\noindent
{\bf Main Result: } There exists a subclass of submodular valuations
that is not GS and (still) guarantees the existence of a Walrasian equilibrium. Furthermore, there exists an ascending auction that always terminates in a Walrasian equilibrium, if all valuations belong to this class.

\vspace{3.5mm}

Recall that a unit-demand valuation assigns arbitrary values to singletons, and the value of every subset of items of cardinality at
least two is equal to the largest value of a singleton in this subset. The new class that we define is a ``unit demand with a
twist'', as follows. Fix some arbitrary number $M$. Then in our valuation class, the value of every subset of items of cardinality at least two is equal to $M$ (and we include all unit demand valuations to which this change does not violate monotonicity and
sub-modularity). This is a very simple class and it is not GS. Nevertheless, a Walrasian equilibrium always exists for any tuple of valuations taken from this class, and it can be obtained via an ascending auction.

Our ascending auction is a careful generalization of the Gul-Stacchetti auction. Although their auction is widely cited, the proof that they give is long and hard to follow. This has been acknowledged several times in the past, for example by Ausubel~\cite{Ausubel2005} who proposes a different auction formulation, with a different analysis. Developing tools and intuition for the main analysis of this paper yields basic structural results for gross substitute valuations, that shed new light on both these auctions and their proofs. First, we show that both these auctions are completely equivalent (while their description is significantly different). More specifically, we show that in every step, both auctions raise the prices of exactly the same items. Second, we give a class of ascending auctions for gross substitutes valuations that contains the Gul-Stacchetti auction as well as other (truly) alternative formats, and that terminate in a Walrasian equilibrium. Our analysis is 
 much simpler (and shorter) than that of Gul and Stacchetti. Since the Gul-Stacchetti auction is an important widely cited result, we feel that a new concise full proof is of independent theoretical interest.

This new family of ascending auctions can be viewed as a family of primal dual algorithms. Bikhchandani and Mamer~\cite{BikhchandaniMa1997} show that a sufficient and necessary condition for the existence of a Walrasian equilibrium is that there will be no integrality gap in a certain linear program, usually termed the ``winner determination linear program''. Ausubel explicitly uses the dual of this program to construct his auction, and by iteratively improving the dual solution he finds Walrasian prices. By incorporating part of his method into the proof of Gul and Stacchetti, we are able to both simplify the proof as well as generalize the result, showing a family of ascending auctions for gross substitutes.

We view our results as the beginning of a search for a full characterization of conditions that ensure both the existence of Walrasian equilibria in subclasses of submodular valuations, and the existence of ascending auctions that find these equilibria outcomes. As our results in this paper demonstrate, this is a complicated technical task. Still, relying on all that we have described above, we believe it is an important task.

\subsection{Additional Related Literature}

\cite{DGS86} were the first to formally define an ascending auction that terminates in a Walrasian equilibrium. They analyzed
the case of unit-demand valuations. \cite{AAT10} fully characterize the class of all ascending auctions that terminate in a Walrasian equilibrium, for unit-demand bidders. It has been later recognized by \cite{GulSt2000} that the DGS auction can be
generalized  to gross substitutes valuations, as detailed above.

The existence of a Walrasian equilibrium for non-GS valuation classes was barely studied. \cite{BikhchandaniMa1997} describe few cases of super-additive valuation classes, and conclude that the exploration of this issue is an important subject for future research. \cite{SY06} identify a class of valuations with complements that ensure the existence of a Walrasian equilibrium, and later provide an iterative auction (that increases and decreases prices) that finds a Walrasian equilibrium for their class of valuations~\cite{SY09}.

Several papers study strategic properties of ascending auctions, showing that truthful demand reporting is a Nash equilibrium in various types of ascending auctions. \cite{A04} shows this for the case of identical items, and more recently \cite{BVSV08} show this when items form a basis of a matroid or a polymatroid.

As mentioned above, it is usually easy to adjust ascending auctions to cases where bidders have budget constraints, as studied by~\cite{DLN08} and more recently by~\cite{Colini-BaldeschiHeLeSt2012,GMP12}.

Finally, we should mention that a strand of the literature that studies ascending auction with nonlinear bundle prices (where there is price for every subset of items), as initiated by Parkes~\cite{parkes99} and by Ausubel and Milgrom~\cite{AM00}.

\section{Preliminaries and Structural Results}\label{sec:prelim}
In the following we denote by $\R_+$ the set of non-negative reals.

A \emph{combinatorial auction} is a setting where a finite set of
items, $\Omega$, is to be
partitioned between a set of players. In what follows it is always
assumed that $|\Omega|=m$, and that there are $n$ players 
associated with valuations $v_i, ~ i=1, \ldots , n$, where $v_i :
2^{\Omega}\mapsto \R_+$. Namely,  for every set $S \subseteq \Omega$,  $v_i(S)$
measures how much the $i$th player favors the set of items $S$.
Valuation are assumed here to be monotone with respect to inclusion
(aka \emph{free disposal}), and that $v_i(\emptyset)=0$ for every
$i$. The auctioneer sets a price for each item. This is given by a
function $p: \Omega \mapsto \R_+$. We refer to $p$ also as a
price-vector. Having a price on single items, the price naturally
extends additively to subsets, namely, for $S \subseteq \Omega,~ p(S)
= \sum_{j \in S} p(j)$.

Having a set of valuations $\{v_i\}_1^n$ and a price vector $p$, the
\emph{utility} of player $i$ from a set $S \subseteq \Omega$,  denoted
by $u_{i,p}(S)$ is defined by $u_{i,p}(S) = v_i(S) - p(S)$. The
\emph{demand} of a player $i$ and price vector $p$ is a collection of
subsets of items of maximum utility. Namely, $D_i(p) = \{S|\ \forall S'\subseteq \Omega, u_{i,p}(S)\geq u_{i,p}(S')\}$. 

Occasionally we omit subscripts when there is no risk of confusion,
e.g., we use $u_i(S)$ instead of $u_{i,p}(S)$ when $p$ is fixed or
well defined.  We use $u_i$ or $u_{i,p}$ to denote the utility of
player $i$, which is its utility from a demand set. So $u_i = u_i(S)$ for some $S\in D_i(p)$ when the price vector is known to be $p$.

An \emph{allocation} of the items is a map $A: \{1, \ldots n\}
\mapsto 2^{\Omega}$, which we also denote by $(S_1,\ldots,
S_n)$. Namely, where $A(i)=S_i \subseteq \Omega$. The requirement is
that the sets $\{S_i\}_{1}^{n}$ are pairwise disjoint. 
We think of an allocation
as allocating each player $i$, the  items of the set $S_i$.
 $S_i$ may be empty, for some $i \in [n] = \{1,2,\ldots,n\}$.

\begin{definition} [Envy Free]
Given a price vector $p$, an allocation is \emph{envy free} w.r.t. $p$
if every player is allocated a demand set. A price vector is envy free if there exists an envy free allocation for it.
\end{definition}

As the name suggests, an envy-free allocation w.r.t. a price vector $p$
signifies an economical situation of 'equilibrium' in the sense that
under the corresponding  price $p$, each player is satisfied with its
allocated set.   Clearly, for every set of
valuations, there is an envy-free price vector and an envy-free
allocation for it - simply take $p$ large enough so that the demand
set of every player contains just the empty set. This fact, and the
discussion below motivates the following definition of a Walrasian
allocation.

A \emph{Walrasian allocation}~\cite{Walras1874} is an envy free allocation for
which every \emph{unallocated item}, has price zero. In other words,
the union of the allocated sets cover all items of positive price. 
 A price vector for which there exists a
Walrasian allocation is called a \emph{Walrasian price vector}. 

The existence of a Walrasian price vector is not guarantied for every
set of valuations. If it exists it is said to be Walrasian equilibrium
and the set of valuation is said to poses a
Walrasian equilibrium. There is another importance of the
Walrasian equilibrium - if it exists,  it is known to maximize the
so called \emph{social welfare} among all
allocations.  The social welfare is $\sum_{i\in
  [n]}{v_i(S_i)}$, the sum of the individual values of the sets that
are allocated. The maximum social welfare can be written as a solution to the
following \emph{winner determination} integer linear program (ILP), $P1$: 

\begin{align}
\max \{\sum_{i,S} v_i(S) \cdot x_{i,S}\}\\
s.t.\ \ 
 \forall{j\in \Omega},\ \ \  & \sum_{i,S|j\in S} x_{i,S}\leq 1\\
 \forall{i\in[n]},\ \ \ & \sum_{S} x_{i,S}\leq 1\\
 \forall{i\in [n],\ S\subseteq \Omega},\ \ \ &\ x_{i,S} \in \{0,1\}
\end{align}
If we relax integrality constraints of the program we get a linear program whose
dual, $P2$, is:
\begin{align}
\min \{\sum_{i}^n \pi_i + \sum_{j\in \Omega} p_j\}\\
s.t.\ \ 
 \forall{i\in [n],S\subseteq \Omega}, \ \ \  & p(S)+\pi_i \geq v_i(S)\\
 \forall{j\in \Omega,\ i\in [n]},\ \ \ & p_j,u_i \geq 0
\end{align}

Bikhchandani and Mamer~\cite{BikhchandaniMa1997} observed that a
Walrasian equilibrium exists iff the value of the maximum social
welfare equals the optimum in the problem $P2$. Namely the integrality
gap of the LP relaxation of  $P1$, is $1$. More over, in this case,
 the set of the optimal dual variables $\{p_j\}_{j \in \Omega}$ is a
 Walrasian price vector. 

By the dual constraints we can bound the dual variables $\pi_i$ from bellow by the players' utility $u_{i,p}$. Since we wish to minimize these values and the utility is defined by the prices, we switch the objective to finding a price vector that minimizes the following.
\begin{definition}[Lyapunov:]
$~\Lyap (p) = \sum_{i} u_{i,p} + \sum_j p_j$
\end{definition}

We usually consider valuations that are rationals. In this case, 
since the (integer) linear program $P1$ is invariant to scaling of
$v_i,~ i\in [n]$, and $P2$ is invariant to scaling of $v_i,~i\in [n]$ and  $p$ by
the same factor, we may assume that
$v_i, i~ \in [n],p,\pi$ are integers when convenient.

We look at the space of functions on $\Omega$ as ordered by the
domination order, namely, for ${p,q}: \Omega \mapsto \R_+$, $p \leq q$
if for every $j \in \Omega$, $p(j) \leq q(j)$. If $p \leq q$ we also
say that $p$ is dominated by $q$.

The \emph{gross substitute} class of valuation is the class in which a player never drops an item whose price was not increased in an ascending auction dynamics. Formally:

\begin{definition}[gross substitute $(GS)$]\label{def:gs}
For two price vectors $p,q$ and a set $S$, let $S^=(p,q) = \{j|\ j\in S, p(j)=q(j)\}$.
A valuation $v$ is \emph{gross substitute} if for every price vector
$p$ and  $S\in D(p)$, for every price vector $q \geq p$, $\exists S'\in
D(q)$ such that $S^=(p,q)\subseteq S'$. 
\end{definition}

In other words, if $S$ is a
demand set for $p$, and $q \geq p$ then there is a demand set for $q$
that contains all elements $i \in S$ for which $p(i)=q(i)$. 

It is known~\cite{GulSt2000} that when valuations are monotone the class gross substitute is equivalent to the class of \emph{single improvement} valuations which is defined by the following.
\begin{definition} [single improvement]
A valuation $v$ is in the class \emph{single improvement} if for every price vector $p$ and a set of items $S\notin D(p)$, there exists a set $T$, such that $u(S)<u(T)$, $|T\setminus S|\leq 1$ and $|S\setminus T|\leq 1$.
\end{definition}

The class of gross substitute valuations is important and has been
extensively studied. It can be showed that gross substitute valuations
are submodular (see formal definition below). It contains additive
valuations (in which every single item has a value, and the value of a
set it the sum of values of its items). For some references on gross
substitute valuations see \cite{Nicholson1998,Browning1999}. In the context of this study, the
importance of gross substitute valuations stems from the following
theorem, and the corresponding ascending-auctions of Gul-Stacchetti
\cite{GulSt2000} and Ausubel \cite{Ausubel2005}.

\begin{theorem}\emph{\cite{KelsoCr1982}}\label{thm:GS-wal}
If a set of valuations $\{v_i\}_1^n$ is gross substitute, then it
posses a Walrasian equilibrium. 
\end{theorem}
Moreover, there is an additional nice feature to it.  Finding a
Walrasian equilibrium is computationally easy if one has the access to
the full representation of the valuations, or even a demand oracle of the valuations, by solving the LP
relaxation to $P1$ above. However, this full knowledge may be too
large to handle, or, and more important, in a real economic situation,
not available. It turns out, however, that there is a natural way to
compute a Walrasian equilibrium for gross substitute valuations, along
with a Walrasian allocation. This is done by a process that is called an 
``ascending-auction''. 
 In an ascending-auction the auctioneer starts from the
price vector $p_0=0$ and at each step $t$, finds whether $p_t$ is
Walrasian, or increases the prices of some elements to obtain the next
price vector $p_{t+1}$. Thus such a process makes sense also as real
economical process.

The notion of ascending-auctions is not well defined. One can
potentially start with $p_0 =0$, compute the Walrasian equilibrium and
Walrasian price vector $p^*$ and just set $p_1 = p^*$ on the very next
step. A natural ascending-auction should be such that the next step
can be decided from the current step with a very limited knowledge of
the individual valuations (e.g., at step $t$ only access to $D_i(p_t),
i=1, \ldots n$ should be used), and the increase in prices should be
'natural' namely, can be shown to be required in order to arrive at an
envy-free allocation. It turns out that for gross substitute valuations
such natural auctions have been proposed \cite{GulSt2000, Ausubel2005}. This will be one of the focuses of the study here.

We end this section with the definition of the class $GGS(k,M)$. 
\begin{definition}[truncation]
  For an integer $k$ and $M \in \R_+$, A function \emph{(}valuation\emph{)} $v:
  2^\Omega \mapsto \R_+$ is a $(k,M)$-truncation of a valuation $u$ if
  $v(S) = u(S)$ for every $S \subseteq \Omega$ for which $|S| < k$, and $v(S) =
  M$ for if $|S| \geq k$.
\end{definition}  
\begin{definition}[$GGS(k,M)$]
Let $k$ be a natural number and $M \in \R_+$. A valuation
$v:2^{\Omega} \mapsto \R_+$ is $GGS(k,M)$, if there is a
gross substitute valuation $g$ on $\Omega$ such that $v$ is the
$(k,M)$-truncation of $g$.

In addition, in order that $v$ will be monotone submodular we further
demand that for every $S$, $g(S) \leq M$ if $|S|< k$ and $g(S)
\geq M$ if $|S|
\geq k.$ 
\end{definition}

It is quite straightforward that $GGS(k,M)$ is submodular for every
$k$ and $M$ as restricted above. Furthermore, it is obvious that any
gross substitute
valuation $g$ is in $GGS(k,M)$ for some $k$ and $M$ - simply, define $k=m+1$,
and $M = \max \{g(S), S \subseteq \Omega \}$. Hence, in this respect,
the union of $GGS(k,M)$ over all $k$ and suitable $M$'s contains
gross substitute. $GGS(1,M)$ is just the constant valuation
function. $GGS(2,M)$ is already quite interesting. A valuation in
this class is arbitrary on singletons, and is $M$ on all pairs. 
  We will show in what follows that $GGS(2,M)$ is not
  gross substitute. Moreover, the 'blocking' type obstacle that is used
  for gross substitute auctions to show  that a price vector is not
  equilibrium (not optimal for the dual LP, $P2$ above), does not hold
  anymore.
 
\subsection{Basic Structural Results for Gross Substitute}
We present here some structural results on gross substitute
valuations. These results will turn useful for the simple proof of an
ascending auction for gross substitute. Some of these results are new,
others might have been known. We include proofs of these properties
in the appendix, for completeness.

For a price vector $p$ we denote by $D_i^*(p)$ the collection of
minimal demand sets, that is, $$D_i^*(p) = \{S|\ S\in D_i(p),\ \forall
T\subset S,\ u_{i,p}(T) < u_{i,p}(S)\}$$

\begin{lemma} [Matroidity of GS]
\label{Structural0}
If $v_i$ is gross substitute then for each $p$, $D^*_i(p)$ forms the
bases of a Matroid. 
\end{lemma}

\begin{definition}\label{def:13}
For a set $S$ and a player $i$, let $f_{i,p}(S) = \min_{D\in D^*_i(p)}{\{|D\cap S|\}}.$

For a set $S\subseteq \Omega$, we denote by $1_S$ the vector which has value $1$ for every element $j\in S$ and $0$ otherwise. When $S=\{j\}$ is a singleton we use also $1_j$.
\end{definition}

\begin{lemma}
\label{Structural2}
Let $p$ be an integer price vector, $S$ a set of items and $p' = p +
1_S$. For a gross substitute player $i$, $u_{i,p} = u_{i,p'} + f_{i,p}(S)$.
\end{lemma}
\begin{definition}
For two price vectors $p, q \in \R_+^n$, $\max(p,q)$ is the vector that in every
coordinate $i\in [n]$, has $\max \{p(i), q(i)\}$. Similarly $\min
(p,q)$ is defined. 
\end{definition}

\begin{lemma} [Lyapunov's sub-modularity]\emph{~\cite{Ausubel2005}}
\label{LSM}
For gross substitute valuations, the Lyapunov is 'submodular'
w.r.t. price vectors and the $\min /~ \max$ operations.
That is, $$\Lyap (\max(p,q)) + \Lyap (\min(p,q)) \leq \Lyap (p) + \Lyap
(q).$$
\end{lemma}
\section{Auctions for Gross Substitute}\label{sec:GS}
We review here the ascending auctions of Gul-Stacchetti
\cite{GulSt2000} and Ausubel \cite{Ausubel2005}. The original proofs for the correctness of
these auctions where long and complicated. We give here a
short proof for this auction.  We also show that although the auctions
in \cite{GulSt2000} and \cite{Ausubel2005} are different in
description, these two auctions are actually the same.

In this section we assume that all valuation are gross substitute
integer functions.
\subsection{Witnesses for non existence of an allocation}\label{sec:3.1}
In the following set of definitions the intention is to
isolate a witness showing that a price vector $p$ does not support an
envy-free allocation. Recall the definition of $f_{i,p}(S)$ from
Definition \ref{def:13}.
\begin{definition}\emph{\cite{GulSt2000}}
Let $p$ be a price vector.
\begin{itemize}
\item For a set $S$, let $f_p(S) = (\sum_{i\in [n]}f_{i,p}(S)) - |S|$.
\item Let $O^*_p\subseteq \Omega$ be
  a set such that $\forall S,\ f_p(S)\leq f_p(O^*_p)$ and $\forall
  S\subset O^*_p,\ f_p(S) < f_p(O^*_p)$. Namely, $O^*_p$ is the set of
  maximum $f_p$ and that is minimal in this respect.

If for all sets $S$, $f_p(S) \leq 0$ we define $O^*_p= \emptyset$.

\end{itemize}
\end{definition}
Again, when the price vector is known we omit the subscript and use
$f_i(S)$, $f(S)$ and $O^*$. The proof of the next observation is in the Appendix.

\begin{observation}\label{obs:obstacle}
If $\exists S$ such that $f(S)>0$ then there is no envy-free
allocation. $\qed$
\end{observation}

Hence, in view of Observation \ref{obs:obstacle}, if no allocation
exists due to a set $S$ with $f(S) > 0$, then $O^*_p$ is the
extremal obstacle for $p$ to support an envy-free allocation. In
particular, the elements in $O^*_p$ are {\em over demanded} in the sense
that in every possible allocation w.r.t $p$, the number of elements in
$O^*_p$ that need to be allocated is more than it size. It thus make
sense to increase the price of these elements, so that they will be
less demanded. This is just what the Gul-Stacchetti auction does. In the
proof of correctness, one will need to also show that when no obstacle
exists, then there is indeed an envy-free allocation, and moreover,
that this is in addition a Walrasian allocation (namely, that every
positive priced item is allocated).  


We state here an additional property (see appendix for proof) of gross substitute valuations.
\begin{lemma} [Monotonicity]
\label{mon}
Let $p' = p + 1_j$, then $\forall S$ for which $ j\notin S,\ f_{p'}(S) \geq f_p(S)$
\end{lemma}

\subsection{Equivalence}

Although Gul and Stacchetti's auction is ascending while Ausubel's can
ascend or descend, we will show their equivalence. 
Namely, if we start with the $0$ price vector, then at
every step the same set of items is chosen for a price
increase in both auctions. 
We next describe both auctions formally.

\vspace{3.5mm}

\textbf{Gul and Stacchetti's auction:}
\begin{itemize}
	\item Start with $p=0$
	\item	While $f_p(O^*_p) \neq \emptyset$ increase $p(j)$ by $1$ for all $j\in O^*$
\end{itemize}

\vspace{3.5mm}

\textbf{Ausubel's auction:}
\begin{itemize}
	\item Start with $p$
	\item Let $\Lyap (\cdot)$ denote the Lyapunov function. While $\exists S,S'$ such that $\Lyap (p+1_S-1_{S'})<
          \Lyap (p)$ find $S^*,S'^*$ such that $\forall T,T',\ \Lyap
          (p+1_{S^*}-1_{S'^*})\leq \Lyap (p+1_T-1_{T'})$ and $\forall
          T\subset S^*,\ T'\subset S'^*$ it holds that $\Lyap
          (p+1_{S^*}-1_{S'^*})< \Lyap (p+1_T-1_{T'})$. Increase $p(j)\ by\ 1$ for all $j\in S^*$ and decrease it by $1$ for all $j\in S'^*$
\end{itemize}

We assume in what follows that the auction
has a Walrasian equilibrium (as every gross substitute does). Note that for Ausubel's auction, if $S,S'$ as above exist it
immediately shows that $p$ cannot be optimal in the dual LP
$P2$. If such $S,S'$ are found, Ausbel's auction improves the Lyapunov
by the suggested change in price. Hence, the auction
is a dual-auction. The actual $S^*, S'^*$ that are chosen are these
that maximize the change in the dual cost $\Lyap$. 
 
It was proved by Ausubel that if the current price vector is not
higher than equilibrium, this process will not decrease any
price. Hence the auction can be described in a simpler way. Define the
\emph{minimal minimizer} set $S^* (=S_p^*)$ to be such that $\forall
S, L(p+1_{S^*})\leq L(p+1_S)$ and if for a set $S\neq S^*$,
$L(p+1_{S^*}) = L(p+1_S)$ then $|S^*| \leq |S|$. Clearly, a minimal
minimizer always exists, but furthermore, Ausubel provides also a
proof of it being unique.\\

\textbf{Ausubel's ascending auction:}
\begin{itemize}
	\item Start with $p=0$
	\item	While $S^* \neq \emptyset$ increase $p(j)$ by $1$ for all $j\in S^*$
\end{itemize}

What we show next (proof in appendix) is that the set $S^*$ in Ausbel's auction collides
with $O^*$ in Gul-Stacchetti's.
\begin{corollary}\label{cor:equiv} \emph{(}of Lemma \ref{Structural2}\emph{)}
The two algorithms collide on gross substitute valuation.
\end{corollary}

We end this section with the remark that  the equivalence shown
above is in direct analog to the divisible goods
scenario. For auctions in the divisible good scenario a similar Lyapunov is used
to identify the set of prices to be changed. The reason is that its sub-gradient at $p$ equals the excess supply, which, in
turn, is the 
analog to the discrete function $f(\cdot)$ (or $-f(\cdot)$ to be
precise). See~\cite{Varian1981} for divisible good combinatorial
auctions. 

\subsection{A finer ascending auction for Gross Substitute}
We conclude this section with a 'finer' version of the Gul-Stacchetti
auction. The idea is the same,
the only difference is that we go by finer steps, namely, at each step
we increase the price of a single element. We also include a full
proof that it reaches the optimal equilibrium. For this we use the
properties that we have stated in Sections
\ref{sec:prelim} and \ref{sec:3.1} with the additional fact that
gross substitute valuations posses Walrasian equilibrium (which is not
proved here).

\vspace{3.5mm}

\textbf{Auction:}
\begin{itemize}
	\item Start with $p=0$
	\item	While $O^* \neq \emptyset$ increase $p(j)$ by $1$ for some $j\in O^*$
\end{itemize}
In order to prove that this algorithm finds an optimal price vector we
will show that it never
goes above any optimal price vector, and when it stops the dual objective value is at-most
the value of the optimal price dual objective. We assume that there
is an optimal Walrasian equilibrium.\\
Let $p^*$ be an optimal price vector, $p_t$ the current price vector and $T$ the last step of the algorithm, that is, we need to show $p^* = p_T$.
The next proposition, which its proof is in the appendix, shows exactly that.

\begin{proposition}\label{upper:bound}
 For every  $t \leq T,$ $p_t\leq p^*$.
\end{proposition}

\begin{claim}
  \label{cl:739}
If for an integral price $p$, $f_p(S) \leq 0$ for some $S$, then
$\Lyap(p+1_S) \geq \Lyap(p)$
\end{claim}
\begin{proof}
This is an immediate corollary of lemma~\ref{Structural2}.
\ignore{ 
 Indeed let $p,S$ as above and $p' = p+1_S$. For the $i$th player
  $u_{i,p'} \geq u_{i,p} - f_{i,p}(S)$. This is so as the demand set
  $D \in D_p(i)$ for which $D \cap S = f_{i,p}(S)$ has decreased its
  utility by $f_{i,p}(S)$ while there might be another set $D' \notin
  D_p(i)$ that decreased its utility by less and becomes a member of
  $D_{p'}(i)$. Hence the difference in the Lyapunov is $\Lyap(p') -
  \Lyap(p) = \sum_i (u_{i,p'} - u_{i,p}) + \sum_{x \in \Omega} p'(x) -
  p(x) \geq |S| -\sum_i f_{i,p}  = -f(S) \geq 0$. }
\end{proof}

\begin{proposition} [Optimality]
$\Lyap (p_T) \leq \Lyap (p^*)$.
\end{proposition}
\begin{proof}
By proposition \ref{upper:bound} $p_T\leq p^*$. Let $O^< = \{j|~
p_T(j)<p^*(j)\}$ be the
elements that have smaller price in $p_T$ than in $p^*$. 
since $T$ is the last step, $f_T(S) \leq 0$ for every set $S$,
and in particular for $f_T(O^<) \leq 0$. Hence by Claim \ref{cl:739} $\Lyap (p_T + 1_{O^<}) \geq \Lyap (p_T)$.

Now, if $p_T +   1_{O^<} = p^*$ this is in
contradiction with the optimality (and minimality) of $p^*$. Otherwise
if $p' = p_T +   1_{O^<} \neq p^*$, the same argument
applies to $p'$. This is true as one can argue that $f_{p'}(S) \leq
0$ for every $S \subseteq O^<$ since this was true for $p_T$ and
increasing prices in $O^<$ by $1$, cannot cause the increase in
$f_{i,p}(S)$ for any set $S \subseteq O^<$. 
\end{proof}

We note that a different proof for the last Lemma is possible along the
following lines: By the Matroid union theorem, it can be shown that
$p_T$ posses an envy-free allocation, 
since $O^*= \emptyset$ w.r.t $p_T$. Proposition \ref{prop:gen_dom}
then shows that this is a Walrasian allocation. However, this uses 
heavier machinery from Matroid theory.

\section{$GGS(k,M)$ valuations}

Recall, a valuation $v$ is $GGS(k,M)$ if there is a gross substitute
valuation $g$ such that $v=g$ for all sets of size less than $k$
- these will be called {\em small sets}, and is equal $M$ for larger
size sets. $GGS(k,M)$ are submodular, but are not gross substitute in
general. Such a small example, showing that $GGS(2,M)$ is already not
gross substitute is in Appendix Section \ref{appen:GS2}. In the
following discussion we will show this, by utilizing the facts we know
about an auction defined on an ensemble of gross substitute
valuations. In particular, this will show that the condition for non
existence of an allocation that is used for gross substitute is not
useful for $GGS(k,M)$ auctions anymore.

We will consider here combinatorial auctions in which all players have
valuations that are in $GGS(k,M)$, namely the same $k,M$ for every
player. Note that this does not mean that all valuations are the
truncation of the same gross substitute
valuation $g$.
 

Recall that for gross substitute valuations, the existence of a
non-empty set $O^*_p$ was a witness of non-optimality of $p$.
Furthermore, $O^*_p$ was the handle in which ascending auctions like
Gul-Stacchetti's and Ausbel's found a set of items whose prices are to
be increased. In the next claim we show that this is not the case for
some valuations in $GGS(2,M)$. In particular, this shows that
$GGS(k,M)$ is not gross substitute already for $k=2$. 

Let $n$ be large enough  and $m=2n-2$. (e.g., $n=5, m=8$ will do).
Let $M=2$ and the following $GGS(2,2)$ valuations:
$v_i, i\leq n-2$ is uniform value
$1$ on all singletons. $v_{n-1}=v_n$  assign value $1$ for 
all items but the last which has value $2$. 

\begin{claim}\label{cl:118}
For the valuations above and $p=0$ there is no allocation, while
$\forall S, f_p(S) \leq 0$.  $\qed$
\end{claim}

\ignore{
As noted, we call sets of $k$ elements or more \emph{big} sets and sets with less than $k$ elements \emph{small} sets. In the same manner, we call players who have big sets in their demand \emph{big players} and players who have small sets in their demand \emph{small players}. Note that a player can be both small and big at the same time.}


\subsection{$GGS(2,M)$ posses  Walrasian equilibrium}\label{sec:ggs(2,m)}

Here we show that $GGS(2,M)$ always posses Walrasian equilibrium. We
don't know whether this fact is true for $GGS(k,M)$ for $k \geq 3$.

\begin{theorem}\label{thm:wal2}
  Let $\Omega$ be a set of $m$ items and let $v_i, ~i=1, \cdots ,n$ be
  valuations on $\Omega$, each in $GGS(2,M)$. Then there is a
  Walrasian equilibrium for $v_1, \ldots, v_n$. 
\end{theorem}
Before we prove Theorem \ref{thm:wal2}, we prove the following
intermediate result.
\begin{lemma}
\label{lem:ev2}
Let $\Omega$ be a set of $m$ items and let $v_i, ~i=1, \cdots ,n$ be
  valuations on $\Omega$, each in $GGS(2,M)$. Let $p$ be a minimal optimal
  price vector \emph{(}namely, $p$ is optimal in the dual LP, $P2$\emph{)}. Then $p$
  poses an envy-free allocation.  
\end{lemma}

We will recurrently use the following standard fact from matching
theory \cite{matching}.
\begin{fact}
  \label{fact1}
$~$

  \begin{enumerate}
  \item Let $G=(V,E)$ be a graph and assume that $V' \subseteq V$ is matched
by some \emph{(}not necessarily maximum\emph{)} matching, then there is a maximum
matching that matches $V'$.
\item Let $G=(X,Y:E)$ be a bipartite graph and $M$ a maximum
  matching. Then \emph{(}K\"onig Egerv\'ary theorem, see e.g.,
  \emph{\cite{matching})} there is a vertex cover $C \subseteq X \cup Y$
  covering the edges of $G$ with $|C| = |M|$. Moreover, $M$ matches
  the $|C \cap X|$ members of $X \cap C$ with items in $Y \setminus C$ and
  it matches the $|C \cap Y|$ members of $Y \cap C$ with members of
  $X\setminus C$.
  \end{enumerate}
\end{fact}

\begin{proof}[Of Lemma \ref{lem:ev2}]

Obviously, to show the existence of an
  envy-free allocation it
  is enough to consider only sets in the demand set of each player
  that are of size at most two (as larger sets don't have larger
  value). Thus a demand set that does not contain the empty set
  typically contains some singletons and sets of size $2$ (referred
  here as pairs). 
Let $p$ be an optimal solution to the associated LP, $P2$ for the
valuations $v_i, ~i=1, \ldots, n$.
In the rest of the proof we partition the players into two sets. The set $S$
of players for which the demand set contains {\em only} singletons, and the
set $P$. Hence the demand set of every player in $P$ contains pairs
too (it might also contain singletons). We also partition the items
into two sets, the set $MIN$ of items of min price, and the set $A$ of
other items. Namely $MIN= \{j \in \Omega|~ \forall j' \in \Omega,
~p(j) \leq p(j')\}$.  

We will recurrently use the following bipartite graph $G=([n],\Omega:
E)$, where $E =\{(i,x)| ~ \{x\} \in D_i(p)\}$. 
\begin{claim}
  \label{cl:at2_1}
If $|MIN| \geq 2$ then for every player $i \in P$, the demand set
$D_i(p)$ contain all pairs in $MIN$ and no other pairs. 
\end{claim}
\begin{proof}
This is obvious as if the minimum price is $p_0$, then the utility of
every pair of items from $MIN$ is $M -2p_0$ while every other
pair has smaller utility.  
\end{proof}
Following similar reasoning we get the following claim.
\begin{claim}\label{cl:min=1}
  Assume that $MIN = \{x\}$, namely $|MIN|=1$.  Let $MIN2 =
\{y\in \Omega| ~\forall z \neq x,~ p(y)\leq p(z) \}$. Namely,
$MIN2$ is the set of items of minimal price in $\Omega\setminus \{x\}$.
 Then for every $i \in P$ its demand pairs are $MIN \times MIN2$. \qed
\end{claim}

\begin{claim}
  \label{cl:at2_2}
There is an allocation of singletons to the players in $S$.
\end{claim}
\begin{proof}
  We only need to
 show that there is a perfect matching in $G$ w.r.t. $S$.

Indeed, if there is no perfect matching then by Hall marriage theorem
there is a set of players $S' \subseteq S$ such that $N(S') = \{x|~
\exists i \in S', ~(i,x) \in E(G) \}$ has cardinality $|N(S')| <
|S'|$. In that case consider an increase of $\epsilon > 0$ in the
price of every element in $N(S')$. The total price will increase by
$\epsilon \cdot |N(S')|$ while the total utility will decrease by at
least $\epsilon \cdot |S'|$. Hence the Lyapunov will decrease in
contradiction with the assumption that $p$ is an optimal price vector.
\end{proof}

\begin{claim}
  \label{cl:at2_3}
If $|MIN| \geq 2$ then there is an allocation of singletons that covers all $A$ and $S$.
\end{claim}
\begin{proof}
We first show that there is a matching $M_A$ in $G$ that matches $A$.  
Then, since Claim
  \ref{cl:at2_2} asserts the existence of a matching $M_S$ that
  matches $S$, it is easy to see that $M_S \cup M_A$ contains a
  matching that covers $A \cup S$. 

Indeed, assume that there is no matching in $G$ that matches $A$. In this
case, again by Hall theorem, there is a set of items $A' \subseteq A$
for which $|N(A')| < |A'|$, where $N(A') = \{i| ~ \exists x \in A',
~(i,x) \in E\}$. Consider the decrease in price by
$\epsilon > 0$ of every element in $A'$. $\epsilon$ is taken small enough so that the elements of minimum price
will not be affected. Then the total price will decrease by $\epsilon
|A'|$. The utility of players from singleton increases only for players
in $N(A')$, and in this case, increases by $\epsilon$. The utility of
pairs does not increase at all (by Claim \ref{cl:at2_1}, as we assume
that $|MIN| \geq 2$). Hence in total the Lyapunov decreases in
contradiction with the optimality.
\end{proof}

We now end the proof of the lemma for the case $|MIN| \geq 2$ by the
following argument. By Claim \ref{cl:at2_3} there is a matching in $G$
that cover $S$ and $A$. Let $M$ be a maximum matching in $G$ that
matches $A$ and $S$. By Fact \ref{fact1} we may assume the existence
of such maximum matching.

If $M$ does not define an envy-free allocation of all players, there
is a set of players $P' \subset P \cup S$ for which $|N(P')| <
|P'|$. It is standard that we may assume in addition that $M$ leaves
unmatched $|P'| - |N(P')|= n'$ players. As $M$ matches $S$, these
unmatched players are in $P$ and in particular, have any pair in $MIN$
in their demand sets. 

Assume $m'$ elements in $MIN$ remain unmatched by $M$. If $2n' \leq
m'$ then the allocation defined by $M$ can be augmented by a
collection of $n'$ disjoint pairs from the unmatched elements in
$MIN$, and we arrive at an envy-free allocation. Hence we may assume
that $2n' > m'$. 

Consider the following increased price vector $\tilde{p}$: for every element in $N(P')$
the price is increased by $2\epsilon$ and for every other element by
$\epsilon>0$ for some small enough $\epsilon$. 

We claim that the Lyapunov has decreased in contradiction with the
optimality of $p$. Indeed, $\Lyap(\tilde{p}) - \Lyap(p) = \sum_{i \in
  [n]} (u_{i,\tilde{p}} - u_{i,p}) + \sum_{x \in \Omega}
(\tilde{p}(x)- p(x))$.  To show that this difference is negative, we
will show that for each matched player (by $M$), its difference in
utility balances off the difference in the price of the element it is
matched to. We then will show that for the other players, the
difference in utility (which is negative), is more in absolute value
than the difference in the price of unmatched items.

Indeed, for every player in $P'$, The difference in the price of its item as
well as the difference in its  utility are $2\epsilon$. This is true
as such player has singletons only  in $N(P')$ (and possibly pair, but
the utility from a pair drops by at least $2\epsilon$). 

For player not in $P'$, obviously the difference in utility
as well as the difference in the price of its matched item, in
absolute value is $\epsilon$.

The total unaccounted difference in price is contributed by the unmatched
items. This is,   by our notations, $\epsilon \cdot m'$. The total
unaccounted difference in utility is contributed by the unmatched
players. Note that all unmatched players are in $P'$. For any such player $x \in P'$ (there are $n'$ such unmatched
players), the change in utility is $2\epsilon$, as singletons have
increased price of $2\epsilon$ while pairs have increased price of at
least $2\epsilon$. Hence, in total, by our assumption that $2n' > m'$
we get that the total Lyapunov has decreased.

For the case where $MIN=\{x\}$, Namely $|MIN|=1$, an analysis similar
in nature to this above shows that there is an envy free allocation. The
detailed argument for this case is presented in the appendix.
\end{proof}

\begin{proof}[of Theorem \ref{thm:wal2}]
Let $p$ be an optimal solution to the associated LP, $P2$ for the
valuations $v_i, i=1, \ldots, n$. Here we assume further that $p$ is
the minimal such optimal price vector (w.r.t. the order on prices as
defined in the preliminaries). 

  We show that $p$ poses a Walrasian allocation rather just an
  envy-free one. We consider the same graph $G$  as in
  the proof of Lemma \ref{lem:ev2}. Assume first that $|MIN| \geq
  2$. In this case Claim \ref{cl:at2_3} asserts that
  there is a matching in $G$ that matches $S \cup A$. Hence, by Fact
  \ref{fact1} there is a maximum matching $M$ in $G$ that matches $S \cup
  A$. We assume that this matching leaves $n'$
  players  and $m'$ items unmatched. Moreover, 
  Lemma \ref{lem:ev2} asserts that there is a envy-free allocation,
  hence $2n' \leq m'$ as the unmatched players are  allocated disjoint
  pairs from $MIN$. 
We denote this set of $m'$ unmatched items by
  $M'$ and the set of $n'$ unmatched players by $N'$.  

To understand the situation, disregarding $M$, we first consider 
 a maximum matching $M_{SA}$ between $S$ and
$A$. By Fact \ref{fact1} we know that there is a cover $C_S \cup C_A$
of all edges going from $S$ to $A$, where $C_S \subseteq S, ~ C_A
\subseteq A$ and $|C_s| + |C_A| = |M_{SA}|$. In addition, it follows
that $M_{SA}$ leaves $r= |A| -(|C_S| + |C_A|)$ items unmatched in $A$,
and $d= |S| - (|C_S) + |C_A|)$ players from $S$. 
Note also that since $C_S \cup C_A$ is a cover, it follows that there
is no edge between $S \setminus (C_S)$ and $A \setminus C_A$.

We now turn our attention back to the matching $M$. $M$ matches $S
\cup A$ (and some other items). It need not be consistent with the
matching $M_{SA}$, but as it matches all players in $S$ and all items
in $A$, it must use extra $r$ players from $P$ and extra $d$ items
from $MIN$.

Now, consider the price vector $\tilde{p}$ that is obtained from $p$ as
follows: We reduce the price of every item in  $A \setminus C_A$ by
$2\epsilon$ and reduce the price of any other item by 
$\epsilon$.
Hence the total decrease in price $-\Delta(\tilde{p}) = p(\Omega) -
\tilde{p}(\Omega) = \epsilon (m' + d + |C_A| + 2|C_S| + 2r )$. 

To analyze the change in utility, note that every player in $P$ has an
increase in its utility by $2\epsilon$. This is also the case for
players in $C_S$. However, note a player in $S \setminus C_S$ has
an increase of $\epsilon$, as its demand contains only singletons from
$C_A \cup MIN$. 
Hence the total increase in
utility is $\Delta(\tilde{u})= \sum_i u_{i,\tilde{p}} - u_{i,p} \leq
2\epsilon n' + 2 \epsilon r + 2\epsilon |C_S| + \epsilon |C_A| +
\epsilon d$. 

Comparing the decrease in price and increase in utility, taking into
account that $2n' \leq m'$ we conclude that $\Delta(\tilde{u}) \leq
-\Delta(\tilde{p})$ which implies that  $\Lyap(\tilde{p})
\leq \Lyap(p)$ which is a contradiction.

In the case $|MIN|=1$, Claim \ref{cl:min=1} implies that all pairs in
the demand sets intersect (in $MIN$). Since there is an envy free
allocation by Lemma \ref{lem:ev2}, it follows that at most one player can
be allocated a pair. We conclude that either $n-1$ items from $A$
are allocated  as singletons, plus one pair, or $n$ items are
allocated as singletons. In any case, if the allocation is not
Walrasian it follows that $|A| \geq
n$ and hence decreasing prices in all items in $A$ by $\epsilon$ will
not increase Lyapunov. 
\end{proof}

\subsection{ An Ascending auction for $GGS(2,k)$}\label{sec:at2_auc}
We propose here an ascending-auction for $GGS(2,k)$. This auction is a
natural generalization of induced gross substitute auction on
unit-demand valuations of which the $GGS(2,k)$ valuations are the
$(2,M)$-truncation. We will prove that this auction
terminates in a Walrasian allocation, and a corresponding
Walrasian equilibrium.

Recall that for $GGS(2,k)$ valuations, and a price vector $p$, we may
restrict the demand sets of each player to its demanded sets of size
at most two. Again, for a price $p$ we will partition the players into
these that have only singletons in their demand sets. These players
will  be called 'small' players. The other players have
pairs in their demand sets (and may also have singletons). We
disregard players whose demand sets contains the empty set.

We note that given a price vector $p$, we may consider the auction
induced only on the small players $S$, and corresponding singletons
$N(S)$. The demand sets for this induced auction are consistent with
the unit-demand valuation (which is gross substitute), and will be
referred as the induced gross substitute auction on $S$. We considered
several gross substitute ascending auctions, in what follows we refer
to the auction of Gul-Stacchetti in which a price of all items in the
over-demanded set $O^*$ is increased at every step.

We assume in what follows, as discussed before, that all valuations
are integral, and that the minimum Walrasian equilibrium price vector
is also integral.

\vspace{3.5mm}

\textbf{Auction:}
\begin{enumerate}
	\item Start with $p=0$
	\item	If the set of small players $S \neq \emptyset$, and
          there is no allocation for this set, 
          increase prices by taking one step of the induced 
          gross substitute auction on these players. Note that as a
          result some small players may cease to be small. Go back to 2.
	\item There is a perfect matching (partial allocation) of
          the small players to items. Find a maximum matching of
          players to singletons that matches all small players.
  Assume that $n'$ players and $m'$ items remained unmatched. 
  \item   \begin{itemize}
    \item If $2 n'\leq m'$ then $p$ is a Walrasian price vector.
    \item  Else ($2 n' > m'$), increase the price of all items in
      $MIN$ by $1$ and go back to
       2. with the new price vector and the empty allocation.
          \end{itemize}

\end{enumerate}

\vspace{3.5mm}

\begin{theorem}
  \label{thm:at2_auc}
Let $v_1, \ldots, v_n$ be $GGS(2,M)$ valuations. Then the auction
above stops at a Walrasian allocation, and a Walrasian equilibrium.
\end{theorem}

\begin{proof}
It is obvious, by the definition of the auction that the price vector
is integral at every step. 
Theorem \ref{thm:wal2} asserts that there is a Walrasian equilibrium.
 Let $p^*$ be an integral Walrasian price vector. Propositions
 \ref{prop:at2_auc1}, \ref{alloc} and \ref{increasMin}
 imply that the proposed auction never increases the prices beyond
 $p^*$, and that there exists an envy free allocation when it stops. 
 Proposition \ref{prop:gen_dom} implies then that the 
 allocation is Walrasian.
\end{proof}
 
\begin{proposition}\label{prop:at2_auc1}
  Let $p$ be a current price vector in the auction above, for which $p
  \leq p^*$.  If there is no allocation of the small players
  w.r.t. $p$, and hence an increase in price is done is step 2., then
  the resulting price vector $p+1_{O^*} \leq p^*$. 
\end{proposition}
\begin{proof}
Recall that the prices in the gross substitute auction that is induced
on the small players increases the price by $1$ for every item in
$O^*_p$ where $O^*_p$ is the minimal most over-demanded set. Now,
assume for the contrary that $p+1_{O^*}$ is not dominated by $p^*$. Then
(using the integrality assumption) 
there must be  a non empty set $O^= \subseteq O^*_p$ of items for which
$\forall i \in O^=,~ p(i) = p^*(i)$. Moreover, $O^* \neq O^=$ as
otherwise, $f_{p^*}(O^*) \geq f_p(O^*) >0$. This is so since 
if $O^* = O^=$ then $\forall i \in O^*, ~p^*(i)= p(i)$, while  all items outside $O^*$ 
have non smaller prices in $p^*$. Hence the demand sets of all
the players that contribute $1$ to $f_p(O^*)$ remain unchanged w.r.t. $p^*$.

Let $B$ the set of players that have demand sets in $O^=$ and that
have no demand set outside $O^*$.  We first claim that $|B| >
|O^=|$. Indeed, otherwise, $f_p(O^* \setminus O^=) \geq f_p(O^*)$
which is in contradiction to either the minimality or extremity of
$O^*$.

We conclude that w.r.t. $p^*$, the set of players $B$ have their
demand sets only in $O^=$ (as the prices of these items did not
increase w.r.t. $p$ while the prices of other items in $O^*$ did
increase). Hence $O^=$ is over-demanded set w.r.t. $p^*$ in
contradiction that $p^*$ is Walrasian. 
\end{proof}

\begin{proposition}
\label{increasMin}
Let $p\leq p^*$ be a current price vector, assume further that 
there is an allocation for the small items, and that price is
increased in step 4. due to the fact that $2n' > m'$, then the
resulting price is still dominated by $p^*$. 
\end{proposition}
\begin{proof}
Assume that $p \leq p^*$ and that contrary to the proposition prices
are to be increased in step 4. resulting in the price $p' = p+1_{MIN}$
which is
not dominated by $p^*$ any more. Hence there is a non empty set $O^=
\subseteq MIN$ of items, for which $\forall i \in O^=,~ p(i) = p^*(i)$. 

We consider several cases here. Assume first that $|O^= \cap MIN| \geq
2$. In this case, all players that had pairs in their demand set
w.r.t. $p$ will still have all pairs in $O^= \cap MIN$ in their demand
set w.r.t. $p^*$ and with the same utility. For small players, since
in $p^*$ they are allocated, their change in utility balances off the
change in price of the items they are matched with. Thus, altogether,
since utility matches off change in price for small players, 
utility has remained the same for non small players, while total
price has increased, it must be the $\Lyap(p^*) \geq \Lyap(p)$ which
is in contradiction to the optimality of $p^*$.

Consider now the case in which $O^= \cap MIN = \{x\}$. Let $MIN2 =
\{y\in \Omega| ~\forall z \neq x,~ p^*(y)\leq p^*(z) \}$. Namely,
$MIN2$ is the set of items of minimal price w.r.t. $p^*$ in $\Omega
\setminus \{x\}$. In this case, w.r.t. $p^*$, all players that are not
small want all pairs in $\{x\} \times MIN2$. Hence, only one such
player can be allocated a pair, and if there is such player, its
utility increased by $\alpha = p^*(y) - p^*(x)$ w.r.t. its utility in
$p$.

Recall also that in $p^*$ there is a Walrasian allocation. Hence, this
allocation assigns singletons to all players possibly except one as
explained above. Consider now the price vector $p^{**}$ obtained from
$p^*$ by decreasing the price of every item except $x$, by
$\alpha$. By the discussion above, the allocation w.r.t. $p^*$ is
still a valid allocation w.r.t. $p^{**}$. Hence $p^{**}$ is a
Walrasian equilibrium in contradiction with the minimality of $p^*$.

\end{proof}
\begin{proposition}
\label{alloc}
If $2n'\leq m'$ then $p$ is a Walrasian price vector.
\end{proposition}
\begin{proof}
In this situation there is an envy-free  allocation (as the $n'$
players can be
assigned arbitrary disjoint pairs from the remaining items). Since by
induction $p \leq p^*$, Proposition \ref{prop:gen_dom} implies that $p$ is Walrasian.
\end{proof}

\begin{proposition}
  \label{prop:gen_dom}
Let $p^*$ be Walrasian equilibrium for an arbitrary combinatorial
auction \emph{(}not necessarily submodular\emph{)}. Let $p \leq p^*$ so that there
is an envy-free allocation with respect to $p$. Then $p$ is
Walrasian equilibrium. 
\end{proposition}
\begin{proof}
  Let $(S_1^*, \ldots ,S_n^*)$ be the Walrasian allocation
  w.r.t. $p^*$ and let $(S_1, \ldots, S_n)$ be the envy-free
  allocation w.r.t. $p$. Let $S = \Omega \setminus \cup_{i=1}^n S_i$
  be the set of unallocated items. We get,

$$\Lyap(p) = p(S) + \sum_i (u_{i,p}(S_i) + p(S_i)) =$$
$$ \sum_i (u_{i,p^*}(S_i) +
p^*(S_i)) - (p^*(S) - p(S)) \leq $$
$$\sum_i u_{i,p^*}(S_i^*) + \sum_i p^*(S_i^*) =
\Lyap(p^*)$$
Where the second equality is since the drop in utility of the set
$S_i$ is just the increase in the prices of the corresponding
items. The first inequality is since $S_i$ is not necessarily in the
demand set of $i$ w.r.t. $p^*$, and since $(p^*(S) - p(S)) \geq 0$. 

Hence, since $\Lyap(p) \leq \Lyap(p^*)$ equality must occur and $p$ is
Walrasian too.
\end{proof}


\section{Further Discussion - Witnesses for the non-existence of an allocation}\label{sec:further}

As already discussed before, a crucial feature of gross substitute
auctions, that allowed an ascending auction is the existence of a
condition that characterize non envy-free prices. Namely, as shown, a
price $p$ does not posses envy-free allocation if and only if there is
a set of over demanded items, namely $O^*_p
\neq \emptyset$. One reason is that assuming that a current non
envy-free price $p$ satisfies $p \leq p^*$ for a Walrasian (or even
just envy-free) $p^*$ directs the auction to what prices should be
increased.

For valuation in $GGS(2,M)$, the condition above was shown not to hold,
but yet, another characterization was developed (although not explicitly
stated). Essentially, this can be read off the proof of Lemma
\ref{lem:ev2}. It states that $p$ is envy-free if and only if there is
a maximum matching of singletons to players that covers all players in $S$,
and leaves $n'$ unmatched players, $m'$ unmatched items from $MIN$
with $2n' \leq m'$.

Hence, a good starting point for the existence of ascending auction
that finds a Walrasian equilibrium if such exists should be the
existence of such characterization. This is also interesting from the
complexity (namely, the computational resources to decide on next step
in an auction) point of view.
It puts the decision
whether a given price vector posses envy-free allocation in $NP \cap
CO-NP$ (here the input is taken as the  minimal demand sets
$D_i^*(p),~ i \in [n]$).

Such a characterization is not expected for any set of valuations. In
fact it can be shown that for some class of valuations, where the
demand sets are restricted to size two (and moreover, each demand set
contains at most two pairs), to decide whether the $0$-price allows
envy-free allocation is NP-hard.
The same result holds for the well known class of fractionally sub-additive
valuations (aka XOS), see~\cite{Feige2006, DobzinskiNiSh2005} for definitions.
However, these classes of allocation are
not submodular, and we do not know if for any submodular class, the
above characterization question is tractable. 

For $GGS(k,M)$ an inefficient characterization exists in a sense, 
(in particular not known to be in $co-NP$ even for
constant $k$). Moreover, we don't know how to use it to
construct an ascending auction that converges to a Walrasian
equilibrium if exist. 
This in turn, raises the following more finer decision problem: For
gross substitute valuations, as well as $GGS(2,M)$, for a given
$p$ that is dominated by some unknown Walrasian price $p^*$, there was
a way to identify an item $x \in \Omega$ for which $p + 1_x$ is still
dominated by $p^*$. We don't know if this can be done for
$GGS(k,M)$. If one could do this, we would immediately get the desired
ascending auction for $GGS(k,M)$.

\ignore{
 It amounts to the
following generalization of the ``over demanded set'' characterization
of gross substitute.

For a set of valuation $v_i : 2^{\Omega}\mapsto \R_+$, a price vector
$p$ and an arbitrary weight vector on elements $w : \Omega \mapsto
\R_+$ let $f_i(w) = \min\{w(S)|~ S \in D_i(p) \}$. Namely, $f_i(w)$ is the
minimal weight of a demand set for player $i$. Note that $f_i(w)$
depends on $p$ too, but this is omitted in the notation for ease of reading.
Let $f(w) = \sum_{i \in [n]} f_i(w)$, and $w(\Omega) = \sum_{x \in
  \Omega} w(x)$ be the total weight of items.

The following proposition is obvious and true for any set of valuations. It
provides a sufficient condition (not always necessary) for non
existence of envy-free allocation for a given price. 
\begin{proposition}\label{cl:87}
  Let $v_i ~ \in [n]$ be a set of valuations, $p$ a price vector, then
  if there a weighting $w$ such that $f(w) - w(\Omega) > 0$ then $p$
  does not posses an envy-free allocation. $\qed$
\end{proposition}
 
Obviously the proposition is correct, since in an  envy-free allocation,
each player gets a demand set of weight at least $f_i(w)$ and hence
$f(w) \leq(\Omega)$. It is also a direct generalization of the
``over-demanded'' set condition, as if $O^*$ is non-empty, then the
weight $w=1_{O^*}$ namely, that is $1$ for every $i \in O^*$ and $0$
elsewhere provides a weighting that exhibits the impossibility of
envy-free allocation.
}

\ignore{
In what follows we assume that $v_i,~ i=1, \ldots ,n$ are $AT(k,M)$
valuations. For any price vector $p$, let $D_i(p)$ be the demand set
of $i$ restricted to sets of size $k$ or smaller.

If $D_i$ contains no $k$-sets we call $i$ a small player. 
\vspace{1cm}

\textbf{Auction for $AT(k,M)$ valuations:}
\begin{enumerate}
	\item Start with $p=0$
	\item	If there is no allocation for the set of small players $S$,
          increase prices by taking one step of the induced 
          gross substitute auction on these players. Note that as a
          result some small players may cease to be small. Go back to 2.
	\item Let $M$ be a perfect allocation of small players into
          small sets. Augment the matched players into a maximum size
          allocation of players to small sets.
  \item Assume that $n'$ players and $m'$ items remained unmatched. 
    \begin{itemize}
    \item If $k n'\leq m'$ then $p$ is a Walrasian price vector.
    \item Else ($k n' > m'$), increase the price of all items in $MIN$
      by $1$ and go to item 2. with the new price vector and the empty
      allocation.
          \end{itemize}
\end{enumerate}

\begin{theorem}
  \label{thm:at2_auc}
  Let $v_1, \ldots, v_n$ be $AT(k,M)$ valuations. Assume further that
  the valuations admit a Walrasian equilibrium. Then the above auction
  stops at a
  Walrasian allocation, and a Walrasian equilibrium. \qed
\end{theorem}
}

\section{Conclusions}

We present the family of submodular valuation classes $GGS(k,M)$ that
generalize gross substitute. We prove that for some non-trivial
subclasses of it ($GGS(2,M)$) Walrasian equilibrium always exists and
there is a natural ascending auction that reaches this 
Walrasian equilibrium.  This is
contrary to the common belief that only gross-substitute valuations
have these property.  

While open questions w.r.t. $GGS(k,M)$ valuations exist and are
implicitly obvious from the discussion above, we reiterate the most
interesting open problem in the context of submodular valuations. This
is to give a characterization, (or at least to find other significant
families) of valuations  inside submodular, but that are not
gross-substitute and that do have Walrasian equilibrium.  An
additional question in this generality (while less well-defined) is
to present such families that admit {\em natural} ascending-auctions.

\bibliographystyle{plain}
\bibliography{biblio}

\newpage
\appendix 
\vspace{0.7cm}

\section{\small {\bf Proofs for claims in Sections 2 and 3.}}

\vspace{0.5cm} In this section we assume that all valuations and
prices are integral. This is w.l.o.g. as taking any set of valuations
and price vector, scaling all by the same constant does not change the
demand sets. As we assume that the valuations are rationals, and by
scaling integral, scaling large enough  means that prices too can
be rounded up to the closest integral value.  More over, as remarked
in Section \ref{sec:prelim} the optimum for the dual LP, $P2$ is also
integral.

Proof of Lemma~\ref{Structural0}
\begin{proof}
  Let $v$ be a gross substitute valuation, $p$ a price vector and let
  $D_1,D_2\in D^*_v(p)$ be two minimal demand sets. 
Let $p'$ be the price vector defined by $p'(j)=p(j)\ \forall
  j\in D_1\cup D_2$ and $p'(j) = \infty$ else. Note that $D_1,D_2$ are
  minimal demand sets on $p'$ as well. Let $j_2\in D_2\setminus D_1$
  and define $p'_{2} = p'+1_{j_2}$. Clearly
  $D_2\notin D(p'_{2})$, since the
  utility of $D_1$ is bigger for this price vector. Note also that
  $u_{p'_{2}}(D_2) = u_p(D_2) -1$.  By the single
  improvement property there exists $D_3$ such that $|D_3\setminus D_2|\leq
  1$, $|D_2\setminus D_3|\leq 1$ and the utility of $D_3$ on
  $p'_{2}$ is grater than $D_2$. Hence $u_{p'_{2}}(D_3) > u_{p'_{2}}(D_2) =
  u_p(D_2) -1$ which implies (by the integrality assumption) that
  $u_p(D_2) \leq u_{p'_{2}}(D_3) \leq u_p(D_3)$. In particular it follows
  that $j_2 \notin D_3, ~D_3 \in D_p$ and (by minimality of
  $D_2$ w.r.t. inclusion) $D_3 \setminus D_2 \neq \emptyset$.

Hence $\exists j_1\in D_1\setminus D_2$ such
  that $j_1\in D_3$. We conclude that $D_3 = D_2 \cup \{j_1\} \setminus \{j_2\}$,
  for $j_1\in D_1\setminus D_2$. Note that this conclusion is exactly
  the replacement property (aka the exchange principle) of a
  Matroid. By induction on $|D_2\setminus D_1|$ this also implies that
  $|D_1|=|D_2|$. 

Formally, the only thing left to prove  is that the set $D^*(p)$ is
not empty. Indeed, for each $p$, the demand set is well defined and hence the
  minimal demand set is not empty.
\end{proof}

Proof of Lemma~\ref{Structural2}
\begin{proof}
Clearly $u_p \leq u_{p'} + f_p(S)$ since a set $D$ for which $|S\cap
D|=f_p = f$,
guaranties a utility $u_{p'}(D)  = u_p(D) -f$. Assume that $u_p < u_{p'} + f$, that is, there exists a set $S'$ that gives a better utility in $p'$. The next claim, implicitly suggests that there exists a base $B$ in $p$, such that $u_{p'}(B) \geq u_{p'}(S')$. Since this holds for every $S'$, and $u_{p'}(B)\leq u_p - f$ the lemma follows.
\end{proof}
\begin{claim} [Utility Distance]
\label{utilDist}
For a gross substitute valuation, let $p$ be a price vector and
$u_p(S) + l = u_p$, that is, the utility of $S$ is smaller than optimum
utility by exactly $l$. Then, there are two sets $R, D$ such that $D
\subseteq S\cup R$, $|R| \leq l$, $D\in D(p)$.
\end{claim}
\begin{proof}
  The proof is by induction on $l$. For $l=0$ the claim holds by $D=S$
  and $R = \emptyset$.  Assume that $l \geq 1$ and $u(S) + l = u_p$.
By the single improvement property there is a set $T$ and an element $j$, such that $u(T) >
  u(S)$ and $T\setminus S \subseteq  \{j\}$. By the
  induction hypothesis, there are two sets $R', D$ such that $D
  \subseteq T\cup R'$, $|R'| \leq u_p - u(T)$, and $D\in D(p)$.  If we now
  take $R = R'\cup \{j\}$ (or $R'=R$ if $T \subseteq S$) we will get
  the two sets $R, D$ as needed, as 
  $D \subseteq S\cup R$, and $|R| \leq u_p - u(T) + 1 \leq u_p - u(S)
  \leq l$.
\end{proof}

Lemma \ref{Structural0} provides information on the demand set at a
price vector $p$. It does not, however, characterize the gross
substitute valuations $v$. In
particular, it does not specify how the demand set $D_i(p)$ changes
when $p$ changes. The next 
lemma go in this later
direction.

\begin{lemma}
\label{Structural1}
Let $p$ be an integer price vector and let $p'=p+1_j$ for some element
$j$, then for a gross substitute valuation $v$, $D^*(p')$ determined
by $D^*(p)$ as follows:
\begin{itemize}
\item If $j$ is not in all members  of
$D^*(p)$, then $D^*(p')$ contains the members of $D^*(p)$ that do not
contain $j$.
\item If $j$ is in all members of $D^*(p)$
then there might be two cases: In the first, $D^*(p') = \{S
\setminus \{j\}|~ S\in
D^*(p)  \}$, namely, $j$ is deleted from any member of
$D^*(p)$, to result in $D^*(p')$. 

In the other case, $D^*(p') = D^*(p) \cup D'$ where $D'$ is a
collection of some new sets, each of the form  $B\cup\{j'\}\setminus
\{j\}$, for an old  $B \in D^*(p)$ and an
element $j'\neq j$.

We note that in Matroid notations, these possibilities are as
follows. In the first case, the Matroid $M(p')$ is result of the
contraction of $j$ from the Matroid $M(p)$.
 
For the 2nd case: In the first possibly, $M(p')$ is the result of
deleting $j$ from $M(p)$. The 2nd possibly is more complicated, in
this case the rank of the Matroid $M(p)$ does not drop, and some new
bases are added.
\end{itemize}  
\end{lemma}

\begin{proof}
  Clearly if for some $B\in D^*(p)$, $j\notin B$ then $B\in D^*(p')$.
  Assume $j$ is in every base of $D^*(p)$. Note that it follows that
  every set in $D(p)$ also contains $j$. The utility of a base went
  down by exactly $1$ with the increase of $j$'s price. In particular,
  by the integrality assumption, 
  every base in $D^*(p)$ is still in $D(p+1_j)$. Let $A$ be a
  new base. Since $j$ is in every demand set for $p$ we know that
  $j\notin A$ and $|A|\leq |B|$ for all bases $B$. More over, it
  follows that $u_p(A) = u_p(B) -1$. 

Since $A \notin D(p)$,
we apply the single improvement property to obtain a set $C$
  with a strictly grater utility and with at most one element more
  than $A$. Since $u_p(A) = u_p(B) -1$ it follows that  $u_p(C) =
  u_p(B)$, namely $C \in D(p)$ and $C\setminus A = \{j\}$.  Now either
  $|A| < |B|$ and hence we conclude that $A = B\setminus \{j\}$ for
  some base $B$. Or $|A| = |B|$ and hence $A =
  B\cup\{j'\}\setminus \{j\}$ for some base $B$ and for every $B \in
  D(p),~ B \in D(p+1_j)$. 
\end{proof}

We now prove Lemma~\ref{LSM}. As a first step we claim
that the following definition is an equivalent definition for functions
 'sub-modularity' in the context above. We note that this in complete
 analogy with the standard relation between submodular functions on the
 Boolean cube and the standard decreasing  marginal property.

\begin{definition}
An integer function $\alpha$ on $p$  has the '\emph{decreasing marginal return}' property if for every $p$ 
and $x,y \in \Omega$,  $$\alpha (p+ 1_x ) +
\alpha(p+ 1_y) \geq \alpha(p) + \alpha(p+  1_x + 1_y)$$.
\end{definition}

\begin{claim}
An integer function $\alpha$ is submodular if and only if $\alpha$ has the decreasing
marginal return property.
\end{claim}

\begin{proof}[Of the claim]
As usual, note that having the decreasing marginal property is a
sub case of sub-modularity. Here we prove the 'other' direction.

Assume that $\alpha$ has the decreasing marginal return property. We prove
that $\alpha$ is submodular on every integer $p,q$ by induction on
$|p-q|_1$. For $|p-q|_1 = 2$, this is just the decreasing marginal
return. Assume then that $|p-q|_1  \geq 3$, and assume w.l.o.g that
$p(x) > q(x)$.
Then, by the induction assumption,
\begin{equation}
\alpha(\max(p-1_x, q)) + \alpha(\min(p,q)) \leq \alpha(p-1_x) + \alpha(q)
\end{equation}
This is so as $|(p-1_x)-q|_1 < |p-q|_1$, and $\min(p,q) = \min(p-1_x, q)$.
Equivalently, we get
$$\alpha(\max(p,q)) - (\alpha(\max(p,q)) - \alpha(\max(p-1_x, q))) +
\alpha(\min(p,q)) $$
$$ \leq
\alpha(p-1_x) + \alpha(q)
$$
However, by the deceasing marginal return, the parenthesis in the left
hand side is less than $\alpha(p) - \alpha(p-1_x)$, (by taking $s=\max(p,q) >
p$). Plugging this to the left hand side implies that
$$ \alpha(\max(p,q)) + \alpha(\min(p,q)) \leq \alpha(p) + \alpha(q)$$
\end{proof}

\begin{proof}[Of Lemma~\ref{LSM}:]
We prove that $\Lyap(\cdot)$ has the decreasing marginal return property. 
Putting the definition in equivalent form, we need to prove that for
every $p$, $\Lyap(p +1_x)-\Lyap(p) \geq \Lyap(p +1_y + 1_x)
- \Lyap(p+1_y)$.

Note first that $\Lyap(p + 1_x) - \Lyap(p) \leq 1$ for every price function
$p$. Namely, by increasing a
price of an item by $1$, the Lyapunov can increase by at most $1$,
since utilities may just decrease. 

Thus, by the observation above, if $\Lyap(p + 1_x)-\Lyap(p)=1$ there is
nothing to prove. We may assume then (by the integrality assumption) that $\Lyap(p + 1_x)-\Lyap(p) =
-r$ for $r \geq 0$. Consider now Lemma \ref{Structural1}. For a player $v_i$ that has in
$D^*_i(p)$ both a base containing $x$ and one that does not contain
$x$, $u_i$, the utility of $i$ does not change by increasing the
price of $x$. Thus if $r \geq 0$, it
must be that at least for one player  $x$ is in every base. Moreover,
by Lemma \ref{Structural1}, for each such player, the utility decreases
by exactly $1$. Hence there must be exactly $r+1$ players for which every
base contains $x$ w.r.t. $p$. Again, by Lemma \ref{Structural1}, for each
such player, moving from $p$ to $p+1_y$ does not change this fact
(as a result, some new players may have the same situation,  namely
that $x$ is in every base w.r.t. $p+ 1_y$, but certainly the old ones
remain). Hence, moving from $p+ 1_y$ to $p+1_y + 1_x$, the utility
goes down by at least $r+1$, and the Lyapunov goes down by at least
$r$ ending the proof.
\end{proof}

Proof of Observation~\ref{obs:obstacle}
\begin{proof}
 In every allocation, the $i$th player  gets at least $f_{i,p}(S)$
 elements from $S$, since every demand set of it intersects $S$ by
 at least this amount. Hence, as an allocation is an
 assignment of disjoint sets, the total number of elements allocated
 from $S$ is $s= \sum_i f_{i,p}(S)$. Thus if $f(S) > 0$, namely $s> |S|$ there is no allocation.
\end{proof}

Proof of Lemma~\ref{mon}
\begin{proof}
  All players that have a base not containing $j$ have less sets to
  select a minimum from.  For players for which the size of members of
  $D^*(p+1_j)$ went down w.r.t. to $D^*(p)$, $f$ did not change since $j\notin S$.  For
  players which the size of members in $D^*(p+1_j)$ did not change,
  (i.e., have new bases where $D^*()$ is viewed as bases of a Matroid)
  , since these new bases are super-sets of the
  old bases without the element $j$, the $f$'s value did not went down
  too.
\end{proof}

We prove the equivalence (corollary~\ref{cor:equiv}) now
\begin{proof}
  For an auction on gross substitute valuations, that starts with
  $p_0=0$, let $O^*_p, S^*_p$ be the corresponding sets in
  Gul-Stacchetti's and Ausubel's auctions. All we need to prove is
  that at every step $O^*_p = S^*_p$.

  Let $p_w$ be a minimal price supporting a Walrasian
  equilibrium. Assume that for a price vector $p \leq p_w$, $O^*_p
  \neq S^*_p$, then since by the definition of Ausbel auction, $S^*$
  is the set that minimize $\Lyap(p + 1_{S^*})$, either $\Lyap
  (p+1_{S^*}) < \Lyap (p+1_{O^*})$, or equality holds. Denote $p_s=
  p+1_{S^*}$, and $p_o = p+1_{O^*}$. In the first
  case, by the Lyapunov definition, $p(\Omega) + |S^*| +
  \sum_i u_{i,p_s} < p(\Omega) + |O^*| +
  \sum_i u_{i,p_o}.$  By Lemma \ref{Structural2} we have that
  $|S^*| - \sum_i{f_{i,p}(S^*)} < |O^*| - \sum_i{f_{i,p}(O^*)}$. This
  implies that $f_p(O^*) < f_p(S^*)$ in a contradiction to the
  maximality of $f(O^*)$. If, on the other hand, $L(p+1_{S^*}) =
  L(p+1_O^*)$ we will get by the same reasoning a contradiction to the
  minimality of $|O^*_p|$, since if $f(O^*) = f(S^*)$, by the 'minimal
  minimizer' definition and uniqueness of $S^*$, it would imply that
  $|S^*|<|O^*|$.
\end{proof}

Proof that the price never goes too high (proposition~\ref{upper:bound})
\begin{proof}
Assume that the statement is false and let $t$ be the last time for which
$p=p_t \leq p^*$.  Let $O^*_t$ be the corresponding obstacle.
If the price of all elements in $O^*_t$ already reached their maximum price,
then by Lemma~[\ref{mon}], $f_{p}(O^*_t) \leq
f_{p^*}(O^*_t)$ a contradiction to the fact that there is an envy-free
allocation w.r.t. $p^*$ (recall Observation \ref{obs:obstacle}).

Hence we may assume that there exists some elements in $O^*_t$ for
which $p$ assigns a smaller price than $p^*$. Denote the set of all
these elements as $O^{\neq}$. Denote the set of the other elements $O^{=}$, that is,
$O^*_t$ is a disjoint union of the  sets $O^{\neq}$ and $O^{=}$.

 Let $p'$ be $p+1_{O^*}$, then $\min(p',p^*) = p+1_{O^{\neq}}$ and
 $\max(p',p^*) = p^*+1_{O^=}$.
 Lemma \ref{LSM} (sub-modularity) implies that $\Lyap (p^*+1_{O^=}) +
\Lyap (p+1_{O^{\neq}}) \leq \Lyap (p') + \Lyap (p^*)$.
By the optimality of $p^*$, $\Lyap (p^*) \leq \Lyap
(p^*+1_{O^=})$ which implies that
$\Lyap (p+1_{O^{\neq}}) \leq \Lyap (p')$. Expanding the definition of
Lyapunov in the last inequality we get, 
$p(\Omega) + |O^{\neq}| + \sum_i u_{i,p+1_{O^{\neq}}} \leq p(\Omega)
  + |O^*| + \sum_i u_{i,p+1_{O^*}}$.

Now by the Lemma \ref{Structural2} we conclude that
$|O^{\neq}| - \sum_i{f_{i,p}(O^{\neq})}\leq |O^*| - \sum_i{f_{i,p}(O^*)}$, hence $f_p(O^*) \leq f_p(O^{\neq})$
which is a contradiction either to the minimality of $O^*$ or to the maximality of $f_p(O^*)$.
\end{proof}

\subsection{$GGS(2,M)$ is not gross substitute}\label{appen:GS2}
For a concrete small example showing that $GGS(2,M)$ is not contained in gross
substitute, consider the following valuation which is in $GGS(2,4)$. There
are three items $a,b,c$ where $v(a)=v(b)=2$, $v(c)=4$. Consider now
the price vectors $p(a,b,c)=(0,1,2)$ and $q(a,b,c)=(2,1,2)$. The set
$\{a,b\}\in D(p)$ since its utility is $3$ which is maximum, but in
$D(q)$ there is no set containing the element $b$ contrary to the
definition of gross substitute valuations (see Definition
\ref{def:gs}).

\subsection{Last case for the proof of Lemma \ref{lem:ev2}}

We present here the proof of the Lemma for the case where
$MIN=\{x\}$, Namely $|MIN|=1$. In this case, let $M$ be a maximum
matching in the graph $G$. We may assume that $M$ is not perfect
w.r.t. $P \cup S$ as otherwise, it would define an envy free
allocation.

By Fact \ref{fact1} there is a cover of the edges of $G$ $C_1 \cup
C_2$, where $C_1 \subseteq (S \cup P), ~ C_2 \subseteq (A \cup
\{x\})$, and $|C_1| + |C_2| =|M|$. Also, since the players in $C_1$
are matched into $(A \cup \{x\} \setminus C_2$, it follows that 
\begin{equation}
  \label{eq:137}
 |C_1| \leq |(A \cup \{x\} \setminus C_2|
\end{equation}

 Moreover, since we assume that $|M|
< |S \cup P|$, it follows that 
\begin{equation}\label{eq:e311}
|C_2| < |(S \cup P) \setminus C_1|
\end{equation}
Assume first that $x \in C_2$. In that case, increasing prices of
items in $C_2$ by small enough $\epsilon$ increases the total price by
$\epsilon |C_2|$ while the utility of every player in $(S \cup P)
\setminus C_1$ decreases by $\epsilon$ (we use here the fact that
utility of pairs decreases by at least $\epsilon$ as $x \in C_2$). By
Equation (\ref{eq:e311}) we conclude that the Lyapunov has decreased. 

Consider now the case where $x \notin C_2$. In this case either $m \leq
n$ (where $m = |\Omega|$), or $m > n$. In the first case the
increase of the price in every item in $A$ by $\epsilon$ obviously decreases
the Lyapunov (since the utility from of every player decreases by
exactly $\epsilon$). In the later case ($m > n$), decreasing the price of
every item in $A$ by $\epsilon$ decreases total price by $\epsilon
(m-1)$, while every player's utility increases by exactly
$\epsilon$. Hence the total Lyapunov does not increase. 
\end{document}